\documentclass[11pt,a4paper,english]{article}
\title{A New Algorithm for   the Inverse of Matrices with Noncommuting Entries}
\usepackage{authblk}
\author[1]{Albert Much\footnote{amuch@matmor.unam.mx}}
\author[2]{Diego Vidal-Cruzprieto \footnote{diego.vidal@correo.nucleares.unam.mx}}
\affil[1]{Centro de Ciencias Matem\'aticas\\
	UNAM\\
	Morelia, Michoac\'an, Mexico}  
\affil[2]{Instituto de Ciencias Nucleares, UNAM,  D.F., M\'exico }    
\usepackage[makeroom]{cancel}
\usepackage[pdftex]{graphicx} 
\usepackage{float}            
\usepackage[T1]{fontenc}   
\usepackage[cmex10]{amsmath}  
\usepackage{amstext}          
\usepackage{mathrsfs}
\usepackage{amsfonts}         
\usepackage{amssymb}          
\usepackage{setspace}
\usepackage{bm}               
\usepackage{enumerate}        
\usepackage[bottom]{footmisc} 
\usepackage{array}            
\usepackage{textcomp}         
\usepackage{pdfpages}         
\usepackage{parskip}          
\usepackage[right]{eurosym}   
\usepackage{xcolor}           
\usepackage[square,numbers]{natbib}	  
\usepackage{amsthm}

\newtheorem{theorem}{\textsc{Theorem}}[section]

\newtheorem{proposition}{\textsc{Proposition}}[section]

\theoremstyle{definition}
\newtheorem{defi}{\textsc{Definition}}[section]

\theoremstyle{remark}
\newtheorem{remark}{Remark}[section]

\usepackage{fancyhdr}
\numberwithin{equation}{section} 
\begin{document}
	\maketitle
	\abstract{By using the quasi-determinant   the construction of the authors in  (\cite{G})  leads  to the inverse of a matrix with noncommuting entries. In this work we offer  a new method  that is more suitable for   physical purposes and motivated by   deformation quantization, where our constructed algorithm  emulates the commutative case and in addition gives corrections coming from the noncommutativity of the entries. Furthermore, we provide an equivalence of the introduced algorithm  and the construction via quasi-determinants.} 
 \tableofcontents
	
	\maketitle
	\section{Introduction} 
	In Feynman's famous work that introduced the path-integral formalism \cite{Fe}, i.e. where a formulation is constructed that is an alternative approach towards nonrelativistic quantum mechanics in addition to  the other two important formulations (the Schr\"odinger and the Heisenberg formulation), an interesting statement   stands out in the introduction: \textit{"The formulation is mathematically equivalent
		to the more usual formulations. There are,
		therefore, no fundamentally new results. However,
		there is a pleasure in recognizing old things
		from a new point of view. Also, there are problems
		for which the new point of view offers a
		distinct advantage."}
	\newline\newline
	Exactly, in this spirit we discuss the results of the current work. Namely, we give   an alternative method, to the one given in \cite{G}, of calculating the inverse of matrices that posses noncommutative entries. Although, we further prove that the formulations are mathematically equivalent, this new viewpoint offers a distinct advantage with respect to physical problems. 
	\\\\
	In particular w.r.t.\ quantum mechanics and certain trials of  quantum gravity, such as noncommutative geometry (see for example   \cite{DFR}, \cite{W}, \cite{A1}, \cite{BM} to mention a few examples)  we  deal with noncommuting elements and from a deformation quantization (see \cite{WA} and references therein) point of view (one approach in noncommutative geometry) an effective approach is guided by the following principle: Formulate   noncommutative frameworks  in terms of   commutative theories  and add to them corrections    that enter in terms of commutators, since commutators contain a deformation parameter (see   \cite{BM, MR1, MR2} for most recent examples). For calculations of  inverses of certain noncommutative  generalizations of commutative geometrical Riemannian entities (for example the metric) the use of the quasi-determinant \`a la \cite{G} is not very helpful in taking care of  orders in the deformation parameter. Therefore, an algorithm that gives the commutative inverse plus perturbations in the deformation parameter (hence taking care of the order of perturbation under investigation) becomes technically helpful for certain effective noncommutative considerations.
	\\\\
	By applying the aforementioned  principles we are able to formulate a new method for calculating the inverses of   matrices with noncommuting entries. We  describe this method briefly; take a $d\times d$-dimensional matrix $\mathbb{A} \in \mathbb{M}_{d}(\mathcal{A})$ with entries in a noncommutative (associative and unital) algebra $\mathcal{A}$. In the next step, take its commutative inverse, i.e. formally the inverse of the matrix if the entries would be commutative, denoted by $\mathbb{A}_c$ and calculate the remainder $$\mathbb{B}=\mathbb{A}\cdot \mathbb{A}_c-\mathbb{I},$$ 
	where $\mathbb{I}$ denotes the unit matrix. It is obvious from the former equation  that the remainder is zero in the case of the entries being commutative.  In the next step write the remainder in terms of the decomposition $$\mathbb{B}=\mathbb{A}\cdot\mathbb{T}.$$ 
	Since we know  explicitly the remainder  $\mathbb{B}$ and the matrix $\mathbb{A}$ we can calculate the left hand side of the decomposition of the remainder, namely $\mathbb{T}$. Then, the inverse of the matrix $\mathbb{A}$ is given by $$\mathbb{A}_c-\mathbb{T}.$$ 
	The introduced algorithm is applied  in this work to calculate the left inverse and  the right inverse respectively for the two dimensional case, i.e. for matrices belonging to $\mathbb{M}_{2}(\mathcal{A})$. The motivation to consider this case  is the fact that in physics two-dimensional models usually serve as  useful toy models. However, the  four-dimensional case (which is interesting in physics due to the spacetime dimension) belongs to a more realistic class and it can     be deduced easily from the two by two case. In  particular using the case introduced in this work one can deduce a formula for matrices in $\mathbb{M}_{2^n}(\mathcal{A})$ for a positive integer  $n$.\newline\newline
	Unexpectedly our construction sheds some light on the discussion w.r.t.\ the appropriate ordering of the determinant, i.e.\  $\det\mathbb{A}=ad-bc$ or the ordering $\det\mathbb{A}=ad-bc$. It is obvious that for a noncommutative division ring those two definitions differ, as was first noted by Caley \cite{Ca}. It turns out that the answer in our framework is straightforward, namely depending on which inverse one constructs (i.e.\ right or left inverse) the chosen ordering gives more manageable decompositions.  
	In particular for the left inverse the  ordering $\det\mathbb{A}=ad-bc$   and for the right inverse the  ordering $\det\mathbb{A}=ad-cb$ is more natural with regards to the decomposition. \newline\newline 
	The work  is organized as follows;   We explicitly calculate the algorithm for two dimensions and give a few propositions that follow from this construction. Moreover, we prove the equivalence to the Gel'fand-method that uses the quasi-determinant.  The last section discusses the different possible definitions of the determinant and their effect on our construction. 
\section{The $\mathbb{M}_{2}(\mathcal{A})$ case}
In this section we apply the algorithm that was sketched in the introduction to the two dimensional case. First we define $\mathcal{A}$  to be   an associative unital algebra $\mathcal{A}$ (see for example \cite[Chapter 2.1]{landi_2014}) over the field of complex numbers $\mathbb{C}$ with a bilinear product $\mathcal{A}\times\mathcal{A}\rightarrow \mathcal{A}$, $\mathcal{A}\times\mathcal{A}\ni(a,b)\rightarrow ab\in\mathcal{A}$, which is distributive over addition   but (in general) not commutative.  Consider a matrix $\mathbb{A} \in \mathbb{M}_{2}(\mathcal{A})$, with noncommutative  valued entries belonging to  $\mathcal{A}$, i.e.\
\begin{align}\label{matrix}
\mathbb{A}=\left(\begin{matrix}
a & b\\ c& d
\end{matrix}\right).
\end{align}
Next, we define the corresponding formal commutative inverses of the matrix $\mathbb{A}$, as those matrices that one would obtain with the common formula of inversion (i.e.\ if the entries of the matrix $\mathbb{A}$ were commutative), where we choose a particular ordering for the determinant that will become clear in the subsequent discussion.  Moreover, we define the so called residues that help us to calculate the inverses  later on.
\begin{defi} \label{res} Let the matrices
  $\,_{c}\mathbb{A}^{-1}_L$ and  $\,_{c}\mathbb{A}^{-1}_R$ be defined as  the left and right formal \textbf{commutative inverses} of  the matrix $\mathbb{A}$ respectively, i.e.\
	\begin{align*}
	\,_{c}\mathbb{A}^{-1}_L={(\det\mathbb{A})}^{-1}\left(\begin{matrix}
	d &-b \\ -c & a
	\end{matrix}\right)\;\;\;\;
	\,_{c}\mathbb{A}^{-1}_R=\left(\begin{matrix}
	d &-b \\ -c & a
	\end{matrix}\right){(\det\mathbb{A})}^{-1},
	\end{align*}
	where $\det\mathbb{A}=ad-cb$. Moreover, let the  matrices $\mathbb{B}_L$ and  $\mathbb{B}_R$ be the so called left or right \textbf{residues} respectively for  the  matrix $\mathbb{A}$ and let them be defined by the following formulas 
		\begin{align*}
		\mathbb{B}_L:=\,_{c}\mathbb{A}^{-1}_L\cdot \mathbb{A}-\mathbb{I}\;\;\;\; \mathbb{B}_R:=\mathbb{A}\cdot \,_{c}\mathbb{A}^{-1}_R-\mathbb{I}.
		\end{align*}
\end{defi}
\begin{remark}
For the rest of this work we  denote the determinant by the symbol $\Delta$, i.e.\
$$\Delta:=\det\mathbb{A}=ad-cb$$
 and  moreover we assume that the determinant is not equal to zero, i.e.\ $\Delta\neq 0.$  In order to obtain the inverse this will be the only condition imposed on the elements of the matrix $\mathbb{A}$. By the choice of the determinant we as well chose an ordering of the elements since another (not equivalent) possibility could be $\det\mathbb{A}=ad-bc$.  We could have defined the determinant using the other order (see Section \ref{s3}) and although the residues differ the inverse  still agrees (as it should) with the one defined by Gel'fand
\end{remark}
 Next, we write the explicit form of the residues which has a nice form since it displays the noncommutativity explicitly by commutators. In particular, it becomes clear that if the elements of the matrix $\mathbb{A}$ are commutative the remainder has to be zero, which fits the narrative of deformation quantization.  
\begin{proposition}\label{residues} For the matrix $\mathbb{A} \in \mathbb{M}_{2}(\mathcal{A})$ (from Equation (\ref{matrix})) the explicit expressions for the left and right residues (see Definition  \ref{res}) in terms of commutators are given as 
	\begin{align*}
	\mathbb{B}_L=\Delta^{-1}\left(\begin{matrix}
	[d,a]-[b,c] & [d,b] \\ [a,c] & 0
	\end{matrix}\right),\;\;\;\;
	\mathbb{B}_R=\left(\begin{matrix}
	[c,b] & [b,a] \\ [c,d] & [d,a]
	\end{matrix}\right)\Delta^{-1}.
	\end{align*}
	\begin{proof}
	From Definition \ref{res} the left residue reads,
		\begin{align*}
		\mathbb{B}_L&=\,_{c}\mathbb{A}^{-1}_L\cdot \mathbb{A}-\mathbb{I}\\&=\Delta^{-1}\left(\begin{matrix}
		d &-b \\ -c & a
		\end{matrix}\right)\cdot\left(\begin{matrix}
		a & b\\ c& d
		\end{matrix}\right)-\mathbb{I}\\
		&=\Delta^{-1}\left(\begin{matrix}
		da-bc & db-bd \\ -ca+ac & -cb+ad
		\end{matrix}\right)-\mathbb{I}\\&=\Delta^{-1}\left(\begin{matrix}
		da-bc & [d,b] \\ [a,c] & -cb+ad
		\end{matrix}\right)-\mathbb{I}\\
		&=\Delta^{-1}\left(\begin{matrix}
		ad+[d,a]-cb-[b,c] & [d,b] \\ [a,c] & -cb+ad
		\end{matrix}\right)-\mathbb{I}\\
		&=\Delta^{-1}\left(\begin{matrix}
		[d,a]-[b,c] & [d,b] \\ [a,c] & 0
		\end{matrix}\right).
		\end{align*} 
		For the right residue Definition \ref{res}  and a straight forward calculation yield
		\begin{align*}
		\mathbb{B}_R&=\mathbb{A}\cdot \,_{c}\mathbb{A}^{-1}_R-\mathbb{I}\\&=\left(\begin{matrix}
		a & b\\ c& d
		\end{matrix}\right)\cdot\left(\begin{matrix}
		d &-b \\ -c & a
		\end{matrix}\right)\Delta^{-1}-\mathbb{I}\\
		&=\left(\begin{matrix}
		ad-bc & -ab+ba \\ cd-dc & -cb+da
		\end{matrix}\right)\Delta^{-1}-\mathbb{I}\\&=\left(\begin{matrix}
		ad-bc & [b,a] \\ [c,d] & da-cb
		\end{matrix}\right)\Delta^{-1}-\mathbb{I}\\
		&=\left(\begin{matrix}
		ad-cb-[b,c] & [b,a] \\ [c,d] & ad+[d,a]-cb
		\end{matrix}\right)	\Delta^{-1}-\mathbb{I}\\
		&=\left(\begin{matrix}
		[c,b] & [b,a] \\ [c,d] & [d,a]
		\end{matrix}\right)\Delta^{-1}.
		\end{align*}
	\end{proof}	
\end{proposition}
Next, we define the decomposition that we use in order to calculate the inverse of the matrix $\mathbb{A}$. We give one for the left and one for the right residue that we defined and gave explicitly above.
\begin{defi}\label{dec}
	The left $\mathbb{T}_L$ or right $\mathbb{T}_R$ matrices are defined as   the  decomposition  that factorize  the residues (left and right respectively) in terms of the matrix $\mathbb{A}$   as follows
	\begin{align}
	\mathbb{B}_L=:\mathbb{T}_L\cdot\mathbb{A},\;\;\;\;\mathbb{B}_R=:\mathbb{A}\cdot\mathbb{T}_R.
	\end{align} 
\end{defi}
Next, we use the residues given in Proposition \ref{residues} and the concrete form of the matrix $\mathbb{A}$ to give the explicit form of the left   and right $\mathbb{T}$- matrices.
\begin{theorem} \label{t1}
	The decomposition matrices  $\mathbb{T}_L$ and $\mathbb{T}_R$ of the residue matrices $\mathbb{B}_L$ and   $\mathbb{B}_R$, given in Definition \ref{dec}, read explicitly 
	\begin{align*}
	\mathbb{T}_L
	&=\Delta^{-1}\left(\begin{matrix}
	 d-\Delta(a-bd^{-1}c)^{-1}& -\left(b+\Delta(c-db^{-1}a)^{-1}\right)\\ [a,c](a-bd^{-1}c)^{-1}  & [a,c](c-db^{-1}a)^{-1} 
	\end{matrix}\right),\\
	\mathbb{T}_R	&=\left(\begin{matrix}
	d-(a-bd^{-1}c)^{-1}\Delta & -b-(c-db^{-1}a)^{-1}\Delta \\
	-c-(b-ac^{-1}d)^{-1}\Delta& a-(d-ca^{-1}b)^{-1}\Delta
	\end{matrix}\right)\Delta^{-1}.	
	\end{align*}
	Moreover, the left decomposition  matrix has an equivalent form in terms of commutators which is given as,
	\begin{align*}
	\mathbb{T}_L
	&=\Delta^{-1}\left(\begin{matrix}
	([b,d]d^{-1}c+[d,a]+[c,b]) (a-bd^{-1}c)^{-1}&  \left([d,a]+[c,b]+[b,d]b^{-1}a\right)(c-db^{-1}a)^{-1}\\ [a,c](a-bd^{-1}c)^{-1}  & [a,c](c-db^{-1}a)^{-1} 
	\end{matrix}\right).
	\end{align*}
	
	\begin{proof}
		Following Definition (\ref{dec}) the left decomposition is
		\begin{align*}
		\Delta^{-1}\left(\begin{matrix}
		[d,a]-[b,c] & [d,b] \\ [a,c] & 0
		\end{matrix}\right)&=\mathbb{T}_L\cdot\mathbb{A}\\&=\left(\begin{matrix}
		L_{11} & L_{12} \\ L_{21} & L_{22}
		\end{matrix}\right)\cdot\left(\begin{matrix}
		a & b \\ c & d
		\end{matrix}\right)\\&=\left(\begin{matrix}
		L_{11}a+L_{12}c & L_{11}b+L_{12}d \\
		L_{21}a+L_{22}c & L_{21}b+L_{22}d
		\end{matrix}\right),
		\end{align*}
		where the subindex $L$ denotes we are solving for the left decomposition. The above defines a system of equations to be solved for $L_{ij}$, where $i,j=1,2$
		\begin{align}
		&\Delta^{-1}\left([d,a]-[b,c]\right)=L_{11}a+L_{12}c \label{fl.1}\\
		&\Delta^{-1}[d,b]=L_{11}b+L_{12}d \label{fl.2}\\
		&\Delta^{-1}[a,c]=L_{21}a+L_{22}c \label{fl.3} \\
		&0=L_{21}b+L_{22}d \label{fl.4}.
		\end{align}
	First we take Equation (\ref{fl.4}) and solve for $L_{22}$ to obtain the explicit expression $L_{22}=-L_{21}bd^{-1}$. Next,  we substitute the former equality into Equation (\ref{fl.3}) and solve for $L_{21}$
		\begin{align*}
		\Delta^{-1}[a,c]=L_{21}(a-bd^{-1}c) 	\end{align*}which leads to 	\begin{align*} L_{21}=\Delta^{-1}[a,c](a-bd^{-1}c)^{-1}.
		\end{align*}
		The former equality for the component  $L_{21}$ is substituted  into the expression for $L_{22}$ and we obtain
		\begin{align*}
		L_{22}=-\Delta^{-1}[a,c](a-bd^{-1}c)^{-1}bd^{-1}.
		\end{align*}
		Next, we take Equation (\ref{fl.2}) and solve for $L_{11}$, which yields $L_{11}=(\Delta^{-1}[d,b]-L_{12}d)b^{-1}$. We substitute the former expression in Equation  (\ref{fl.1}) and solve for $L_{12}$,
		\begin{align*}
		\Delta^{-1}\left([d,a]-[b,c]\right)&=\Delta^{-1}[d,b]b^{-1}a+L_{12}(c-db^{-1}a)	\end{align*}by basic operations we have $L_{12}=\Delta^{-1}\left([d,a]-[b,c]-[d,b]b^{-1}a\right)(c-db^{-1}a)^{-1}$ that as well can be written as \begin{align*}
		L_{12}  
		&=-\Delta^{-1}\left(bc-bdb^{-1}a+\Delta\right)(c-db^{-1}a)^{-1}\\
		&=-\Delta^{-1}b-(c-db^{-1}a)^{-1}\\
		&=-\Delta^{-1}\left(b+\Delta(c-db^{-1}a)^{-1}\right),
		\end{align*}
		which we substitute in the expression we had for $L_{11}$ to obtain
		\begin{align*}
		L_{11}&=\Delta^{-1}([d,b]d^{-1}+b+\Delta(c-db^{-1}a)^{-1})db^{-1}
	\\	&=
		\Delta^{-1}(dbd^{-1}+\Delta(c-db^{-1}a)^{-1})db^{-1}\\&=
		\Delta^{-1}(d+\Delta(bd^{-1}c- a)^{-1}),		\end{align*}
		that has the equivalent form in commutators
		\begin{align*}
	L_{11}&=	\Delta^{-1}(d(bd^{-1}c- a)+\Delta) (bd^{-1}c- a)^{-1}\\&=
		-\Delta^{-1}([b,d]d^{-1}c+[d,a]+[c,b])  (bd^{-1}c- a)^{-1}
		.\end{align*}
	Next, we explicitly calculate the right decomposition by making use of Definition (\ref{dec}), 
		\begin{align*}
		\left(\begin{matrix}
		[c,b] & [b,a] \\ [c,d] & [d,a]
		\end{matrix}\right)\Delta^{-1}&=\mathbb{A}\cdot\mathbb{T}_R\\&=\left(\begin{matrix}
		a & b \\ c & d
		\end{matrix}\right)\cdot \left(\begin{matrix}
		R_{11} & R_{12} \\ R_{21} & R_{22}
		\end{matrix}\right)\\&=\left(\begin{matrix}
		aR_{11}+bR_{21} & aR_{12}+bR_{22} \\ cR_{11}+dR_{21} & cR_{12}+dR_{22}
		\end{matrix}\right),
		\end{align*}
		as in the last case we solve for $R_{ij}$.
		\begin{align}
		[c,b]\Delta^{-1}&=aR_{11}+bR_{21}\label{f.1} \\
		[b,a]\Delta^{-1}&=aR_{12}+bR_{22}\label{f.2} \\ [c,d]\Delta^{-1}&=cR_{11}+dR_{21}\label{f.3} \\ [d,a]\Delta^{-1}&=cR_{12}+dR_{22} \label{f.4}.
		\end{align} 
		We start by solving for $R_{11}$, by taking Equation (\ref{f.1}) and subtracting $bR_{21}$ and then multiplying by $a^{-1}$ from the left to  obtain $R_{11}=a^{-1}([c,b]\Delta^{-1}-bR_{21})$. We substitute the last expression back into Equation (\ref{f.3})
		\begin{align*}
		[c,d]\Delta^{-1}&=ca^{-1}([c,b]\Delta^{-1}-bR_{21})+dR_{21}\\
		& =ca^{-1}[c,b]\Delta^{-1}+(d-ca^{-1}b)R_{21} 
			\end{align*}
 and then solve for the component $R_{21}$, where we have
			\begin{align*}
	R_{21}	&=(d-ca^{-1}b)^{-1}([c,d]-ca^{-1}[c,b])\Delta^{-1}\\
		&=(ac^{-1}d-b)^{-1}(ac^{-1}[c,d]-[c,b])\Delta^{-1}\\
		&=((ac^{-1}d-b)^{-1}\Delta-c)\Delta^{-1},
		\end{align*}
		substituting the value for $R_{21}$ in the expression we had for $R_{11}$ leads us to
		\begin{align*}
		R_{11}&=a^{-1}([c,b]\Delta^{-1}-bR_{21})\\&
		=a^{-1}([c,b]-b((ac^{-1}d-b)^{-1}\Delta-c))\Delta^{-1}\\
		&=a^{-1}(cb-b(ac^{-1}d-b)^{-1}\Delta)\Delta^{-1} \\
		&=a^{-1}(cb+(ad-ad)-b(ac^{-1}d-b)^{-1}\Delta)\Delta^{-1} \\
		&=a^{-1}(cb  -ad -b(ac^{-1}d-b)^{-1}\Delta)\Delta^{-1} +d\Delta^{-1}\\
		&=a^{-1}(-(ac^{-1}d-b) -b )(ac^{-1}d-b)^{-1}  +d\Delta^{-1}\\
		&=-c^{-1}d(ac^{-1}d-b)^{-1}  +d\Delta^{-1}\\
		&=-(a-bd^{-1}c)^{-1}  +d\Delta^{-1}\\
		\end{align*} 
		Next, we take Equation (\ref{f.2}) and solve for $R_{12}$, which gives $R_{12}=a^{-1}([b,a]\Delta^{-1}-bR_{22})$ that we substitute   into Equation (\ref{f.4}) 
		\begin{align*}
		[d,a]\Delta^{-1}=ca^{-1}([b,a]\Delta^{-1}-bR_{22})+dR_{22}
				\end{align*}
		and solve for $R_{22}$
				\begin{align*}
		R_{22}&=(d-ca^{-1}b)^{-1}([d,a]-ca^{-1}[b,a])\Delta^{-1}\\
		&=(a-(d-ca^{-1}b)^{-1}\Delta)\Delta^{-1},
		\end{align*}
		which further allows us to get the explicit form of $R_{12}$
		\begin{align*}
		R_{12}&=a^{-1}([b,a]-b(a-(d-ca^{-1}b)^{-1}\Delta))\Delta^{-1}\\
		&=(-b+(db^{-1}a-c)^{-1}\Delta)\Delta^{-1}.
		\end{align*}
	\end{proof}
\end{theorem}
By using the explicit expression for the decomposition matrices $\mathbb{T}_L$ and $\mathbb{T}_R$ the inverses are given in the following theorem. 
\begin{theorem}[Inverse]\label{inv} Let the     matrices $\mathbb{T}_L$ and $\mathbb{T}_R$ be the decompositions given in Theorem \ref{t1}. Then, the  left and right   inverse matrices of the matrix $\mathbb{A}$  are given by the following expressions
	\begin{align}\label{rinv}
	\mathbb{A}^{-1}_L =\,_{c}\mathbb{A}^{-1}_L-\mathbb{T}_L,\;\;\;\;\mathbb{A}^{-1}_R =\,_{c}\mathbb{A}^{-1}_R-\mathbb{T}_R,
	\end{align}  
	where the left inverse  and the   right inverse matrices	   are equivalent to the  inverse   defined via quasideterminants  of  Gel'fand denoted by $\mathbb{A}_\mathcal{G}^{-1}$ (see \cite[Proposition 1.2.1.]{G}),
	\begin{align*}
		\mathbb{A}^{-1}_L=	\mathbb{A}^{-1}_R=\mathbb{A}_\mathcal{G}^{-1}=&\left(\begin{matrix}
	(a-bd^{-1}c)^{-1} & -(db^{-1}a-c)^{-1}\\ 
	-(ac^{-1}d-b)^{-1} & (d-ca^{-1}b)^{-1}
	\end{matrix}\right).\label{gel}
	\end{align*}\end{theorem}
	\begin{proof}
		The left inverse gives the identity when acting by multiplication from the left on $\mathbb{A}$, while the right inverse gives the identity when acting by multiplication from the right on $\mathbb{A}$, which can be proven by a simple matrix multiplication, i.e.\
		\begin{align*}
		\mathbb{A}^{-1}_L\cdot\mathbb{A}=\left(\,_{c}\mathbb{A}^{-1}_L-\mathbb{T}_L\right)\cdot\mathbb{A}=\,_{c}\mathbb{A}^{-1}_L\cdot\mathbb{A}-\underbrace{\mathbb{T}_L\cdot\mathbb{A}}_{=\mathbb{B}_L}=\,_{c}\mathbb{A}^{-1}_L\cdot\mathbb{A}-\,_{c}\mathbb{A}^{-1}_L\cdot\mathbb{A}+\mathbb{I}=\mathbb{I},\\
		\mathbb{A}\cdot\mathbb{A}^{-1}_R=\mathbb{A}\cdot\left(\,_{c}\mathbb{A}^{-1}_R-\mathbb{T}_R\right)=\mathbb{A}\cdot\,_{c}\mathbb{A}^{-1}_R-\underbrace{\mathbb{A}\cdot\mathbb{T}_R}_{=\mathbb{B}_R}=\mathbb{A}\cdot\,_{c}\mathbb{A}^{-1}_R-\mathbb{A}\cdot\,_{c}\mathbb{A}^{-1}_R+\mathbb{I}=\mathbb{I}.
		\end{align*} 
		From the Equality (\ref{inv}) and from the explicit form of the left decomposition matrix (see Theorem \ref{t1}) it is easy to see that the left inverse matrix has explicit expression form 
		\begin{align} 
	\mathbb{A}^{-1}_L=\Delta^{-1}\left(\begin{matrix}
	-\Delta(bd^{-1}c-a)^{-1}& \Delta(c-db^{-1}a)^{-1}\\ -c-[a,c](a-bd^{-1}c)^{-1}  & a-[c,a](db^{-1}a-c)^{-1} 
	\end{matrix}\right).
	\end{align}
		 Next, we prove that the left-inverse matrix is the same as the Gel'fand inverse. For the components $(\mathbb{A}^{-1}_L)_{11}$ and $(\mathbb{A}^{-1}_L)_{12}$ it is obvious (by multiplying the inverse of the  determinant and by extracting a minus sign).  For the component $(\mathbb{A}^{-1}_L)_{21}$ we have,
		 \begin{align*}
		-\Delta(c+[a,c](a-bd^{-1}c)^{-1})&=-\Delta(c(a-bd^{-1}c)+[a,c])(a-bd^{-1}c)^{-1}\\&= -\Delta \,(-  cbd^{-1} +a)c(a-bd^{-1}c)^{-1}
		\\&= -\Delta \,(  -cb +ad)d^{-1}c(a-bd^{-1}c)^{-1}	\\&=
	-	d^{-1}c(a-bd^{-1}c)^{-1}	\\&=-(ac^{-1}d-b)^{-1},
		 \end{align*}
		 where in the last line we pulled the term $d^{-1}c$ into the inverse. For the remaining component $(\mathbb{A}^{-1}_L)_{22}$ we have 
		 		 \begin{align*}
		 		\Delta (a-[c,a](db^{-1}a-c)^{-1} )&=	\Delta (a(db^{-1}a-c)-[c,a] )(db^{-1}a-c)^{-1}\\&=
		 			\Delta (adb^{-1} -c  )a(db^{-1}a-c)^{-1}\\&=
		 			\Delta (ad -cb  )b^{-1}a(db^{-1}a-c)^{-1}\\&=
		 			 b^{-1}a(db^{-1}a-c)^{-1}\\&=(d-ca^{-1}b)^{-1}
		 		 \end{align*}
		While for the right inverse we use in an analogous fashion  Equality (\ref{inv}) and the explicit form of the right decomposition matrix (see Theorem \ref{t1}) and obtain
		\begin{align*}
		\mathbb{A}^{-1}_R&=\left(\begin{matrix}
		d &-b \\ -c & a
		\end{matrix}\right)\Delta^{-1}-\left(\begin{matrix}
		d-(a-bd^{-1}c)^{-1}\Delta & -b+(db^{-1}a-c)^{-1}\Delta \\
		(ac^{-1}d-b)^{-1}\Delta-c & a-(d-ca^{-1}b)^{-1}\Delta
		\end{matrix}\right)\Delta^{-1}\\
		&=\left(\begin{matrix}
		(a-bd^{-1}c)^{-1} &-(db^{-1}a-c)^{-1} \\ -(ac^{-1}d-b)^{-1} & (d-ca^{-1}b)^{-1}
		\end{matrix}\right),
		\end{align*}
		This shows that the our right inverse $\mathbb{A}^{-1}_R$ (as the left inverse) is equivalent to the inverse introduced by Gel'fand. 
	\end{proof}
 \section{The Aftermath}\label{s3}
 	An important question with regards to  our results and in particular to our algorithm is the following: How does changing the definition of the determinant (i.e. choosing another ordering)  reproduce  the result obtained by Gel'fand and how does it change the formerly obtained results?  This section is committed  to  answer those questions.
\begin{proposition}Let  the determinant be defined by a different ordering, i.e.: $\Delta'=ad-bc$ and let $\,_{c}\mathbb{A'}^{-1}_R$ denote the right commutative inverse w.r.t.\ that  ordering, i.e.
		\begin{align*}
	\,_{c}\mathbb{A'}^{-1}_R=\left(\begin{matrix}
	d &-b \\ -c & a
	\end{matrix}\right){\Delta'}^{-1},
	\end{align*}
	 Then, the  \textbf{right residue matrix} given by $\mathbb{A}\cdot \,_{c}\mathbb{A'}^{-1}_R-\mathbb{I}$ is explicitly given by
	\begin{align*}
	\mathbb{B'}_R&=\begin{pmatrix}
	0 & [b,a] \\ [c,d] & [d,a]-[c,b]
	\end{pmatrix}(\Delta')^{-1}.
	\end{align*}	
	and the \textbf{right decomposition matrix} $\mathbb{T'}_R$ w.r.t.\ the different determinant ordering extracted from decomposition $\mathbb{B'}_R=\mathbb{A}\cdot\mathbb{T}_R$ is given explicitly by
		\begin{align*}
	\mathbb{T'}_R&=\begin{pmatrix}
	(db^{-1}a-c)^{-1}[d,c]  &(db^{-1}a-c)^{-1}\left( [a,d]+[c,b]+db^{-1}[b,a]\right) \\ (d-ca^{-1}b)^{-1}[c,d]  &( ca^{-1}b-d)^{-1}\left([a,d]+[c,b]+ca^{-1}[b,a]\right)
	\end{pmatrix}(\Delta')^{-1}.
	\end{align*}
	The inverse matrix obtained by this ordering, i.e.\ by the equality $\mathbb{A'}^{-1}_R=\,_{c}\mathbb{A'}^{-1}_R-\mathbb{T'}_R$ is equivalent to the inverse obtained by the ordering $\Delta=ad-bc$ and therefore equivalent to the Gel'fand inverse, 
	\begin{align*}
\mathbb{A'}^{-1}_R=\mathbb{A}^{-1}_R=\mathbb{A}_\mathcal{G}^{-1}.
	\end{align*}
\end{proposition}
 \begin{proof} The explicit form of the right residue matrix w.r.t.\ to the different ordering can obtained by a straight forward calculation. Using the explicit form of this matrix  we  calculate the {right decomposition matrix} via the following system of equations to be solved for $R_{ij}$, for $i,j=1,2$
 	\begin{align} 
 	aR_{11}+bR_{21}&=0 \\\label{k.2}
 	aR_{12}+bR_{22}&=[b,a](\Delta')^{-1} \\\label{k.3}
 	cR_{11}+dR_{21}&=[c,d](\Delta')^{-1} \\\label{k.4}
 	cR_{12}+dR_{22}&=([d,a]-[c,b])(\Delta')^{-1}. 
 	\end{align}
 Using the first Equation   we observe that $R_{11}=-a^{-1}bR_{21}$, substituting it into Equation (\ref{k.3}) gives us $R_{21}=(d-ca^{-1}b)^{-1}[c,d](\Delta')^{-1}$, therefore the final value for $R_{11}$ is expressed by
 \begin{align*}
   R_{11}&=-a^{-1}b(d-ca^{-1}b)^{-1}[c,d](\Delta')^{-1} \\&=
   - (db^{-1}a-c)^{-1}[c,d](\Delta')^{-1}.
 \end{align*} Next, we turn our attention to the remaining variables and proceed in the same fashion. First we take Equation (\ref{k.2}) and solve for $R_{12}$ which leads to $R_{12}=a^{-1}([b,a](\Delta')^{-1}-bR_{22})$, substituting this into Equation (\ref{k.4}) we obtain
 	\begin{align*}
 	R_{22}&=(d-ca^{-1}b)^{-1}\left([d,a]-[c,b]-ca^{-1}[b,a]\right)(\Delta')^{-1}\\
 	&=\left(a-(d-ca^{-1}b)^{-1}\Delta'\right)(\Delta')^{-1}.
 	\end{align*}
 	Thus, we find that $R_{12}$ is given by
 	\begin{align*} 
 	R_{12} &=a^{-1}([b,a] -b \left(d-ca^{-1}b)^{-1}\left([d,a]-[c,b]-ca^{-1}[b,a]\right)\right) (\Delta')^{-1}
 	\\&=a^{-1}b(d-ca^{-1}b)^{-1}\left((d-ca^{-1}b)b^{-1}[b,a] - \left([d,a]-[c,b]-ca^{-1}[b,a]\right)\right) (\Delta')^{-1}\\&=a^{-1}b(d-ca^{-1}b)^{-1}\left( [a,d]+[c,b]+db^{-1}[b,a]\right) (\Delta')^{-1}\\&= (db^{-1}a-c)^{-1}\left( [a,d]+[c,b]+db^{-1}[b,a]\right) (\Delta')^{-1}
 	\end{align*}
 	The former term can also be written as 
 	\begin{align*} 
 	R_{12} &=(db^{-1}a-c)^{-1}\left(ad-da+cb-bc+da-db^{-1}ab\right) (\Delta')^{-1}\\& 
 	=(db^{-1}a-c)^{-1}\left(ad-bc  +cb-db^{-1}ab\right) (\Delta')^{-1}\\& 
 	=(db^{-1}a-c)^{-1}\left(\Delta'  +(c -db^{-1}a)b\right) (\Delta')^{-1}
 	\\&
 	=\left(-b+ (db^{-1}a-c)^{-1}\Delta'\right)(\Delta')^{-1}
 	\end{align*}
 	where the form given in the last equality is technically easier to handle in the following calculations.  	Since we have a decomposition, we can calculate the \textbf{right inverse matrix} by Theorem \ref{inv} where the commutative inverse is defined using the determinant $\Delta'$. Next, we prove component-wise that with the different ordering of the determinant  we  reproduce, as before, the result obtained by Gel'fand. For the first component we have
 	\begin{align*}
 	(\mathbb{A}^{-1}_R)_{11}&= (d+(db^{-1}a-c)^{-1}[c,d])(\Delta')^{-1} 
 \\&=(db^{-1}a-c)^{-1}((db^{-1}a-c)d+[c,d])(\Delta')^{-1}\\&=(db^{-1}a-c)^{-1}d\underbrace{(b^{-1}ad-c)}_{=b^{-1}\Delta'}(\Delta')^{-1}\\
 	&=(db^{-1}a-c)^{-1}db^{-1}=(a-bd^{-1}c)^{-1}.
 	\end{align*} 
 	For both $(\mathbb{A}^{-1}_R)_{12}$ and $(\mathbb{A}^{-1}_R)_{22}$ the proof is immediate, for the sake of brevity we will not include it here. The third component renders
 	\begin{align*}
 (	\mathbb{A}^{-1}_R)_{21}&=-\left(c+(d-ca^{-1}b)^{-1}[c,d]\right)(\Delta')^{-1}\\&=-(d-ca^{-1}b)^{-1}\left((d-ca^{-1}b)c+[c,d]\right)(\Delta')^{-1}\\
 	&=-(d-ca^{-1}b)^{-1}c\left(d-a^{-1}bc\right)(\Delta')^{-1}\\&=-(d-ca^{-1}b)^{-1}ca^{-1}\\&=-(ac^{-1}d-b)^{-1}.
 	\end{align*}
 	 \end{proof}
 	Following the same train of thought we give a similar proposition w.r.t.\ the \textbf{left inverse matrix}.
 	\begin{proposition}Let  the determinant be defined by a different ordering, i.e.: $\Delta'=ad-bc$ and let $\,_{c}\mathbb{A'}^{-1}_L$ denote the left commutative inverse w.r.t.\ that  ordering, i.e.
 		\begin{align*}
 		\,_{c}\mathbb{A'}^{-1}_L={\Delta'}^{-1}\left(\begin{matrix}
 		d &-b \\ -c & a
 		\end{matrix}\right),
 		\end{align*}
 		Then, the  \textbf{left residue matrix} given by $\,_{c}\mathbb{A'}^{-1}_L\cdot\mathbb{A} -\mathbb{I}$ is explicitly given by
 			\begin{align*}
 		\mathbb{B}_L&=(\Delta')^{-1}\left(\begin{matrix}
 		[d,a] & [d,b] \\ [a,c] & [b,c]
 		\end{matrix}\right),
 		\end{align*} 
 		and the \textbf{left decomposition matrix} $\mathbb{T'}_L$ w.r.t.\ the different determinant ordering extracted from decomposition $\mathbb{B'}_L=\mathbb{T}_L\cdot\mathbb{A}$ is given explicitly by
 		\begin{align*}
 		\mathbb{T'}_L&=(\Delta')^{-1}\begin{pmatrix} 
 	 d- \Delta' (a-bd^{-1}c)^{-1}    &([d,b]-[d,a]a^{-1}b)(d-ca^{-1}b)^{-1} \\  -c+ \Delta' (ac^{-1}d-b)^{-1}   &([b,c]-[a,c]a^{-1}b)(d-ca^{-1}b)^{-1}
 		\end{pmatrix}.
 		\end{align*}
 		The inverse matrix obtained by this ordering, i.e.\ by the equality $\mathbb{A'}^{-1}_L=\,_{c}\mathbb{A'}^{-1}_L-\mathbb{T'}_L$ is equivalent to the inverse obtained by the ordering $\Delta=ad-bc$,
 		\begin{align*}
 		\mathbb{A'}^{-1}_L=\mathbb{A}^{-1}_L=\mathbb{A}_\mathcal{G}^{-1}.
 		\end{align*}
 	\end{proposition}
 	\begin{proof}The explicit form of the right residue matrix w.r.t.\ to the different ordering can obtained, analogously to the former proof, by a straight forward calculation. By using the explicit form of the residue matrix we obtain the following associated system of equations in order to calculate the \textbf{left decomposition matrix} w.r.t.\ the different ordering 
 	\begin{align}
 		&(\Delta')^{-1}[d,a]=L_{11}a+L_{12}c \label{l.1}\\
 		&(\Delta')^{-1}[d,b]=L_{11}b+L_{12}d \label{l.2}\\
 		&(\Delta')^{-1}[a,c]=L_{21}a+L_{22}c \label{l.3} \\
 		&(\Delta')^{-1}[b,c]=L_{21}b+L_{22}d \label{l.4}.
 	\end{align} 
 	Taking Equation (\ref{l.1}) and solving for $L_{11}$ yields $L_{11}=((\Delta')^{-1}[d,a]-L_{12}c)a^{-1}$, substituting this into Equation (\ref{l.2}) gives us $L_{12}=(\Delta')^{-1}([d,b]-[d,a]a^{-1}b)(d-ca^{-1}b)^{-1}$ and therefore
 	\begin{align*}
 		L_{11}&=(\Delta')^{-1}([d,a]a^{-1}(ac^{-1}d-b)-([d,b]-[d,a]a^{-1}b))(ac^{-1}d-b)^{-1}\\
 		&=(\Delta')^{-1}([d,a]c^{-1}d-[d,b])(ac^{-1}d-b)^{-1}\\
 		&=(\Delta')^{-1}(d(a-bd^{-1}c)-\Delta' )(a-bd^{-1}c)^{-1}\\
 		&=((\Delta')^{-1}d-(a-bd^{-1}c)^{-1} ).
 	\end{align*}
 	For the two remaining variables we take Equation (\ref{l.3}) and solve for $L_{21}$, obtaining $L_{21}=((\Delta')^{-1}[a,c]-L_{22}c )a^{-1}$, which substituted into Equation (\ref{l.4}) yields $L_{22}=(\Delta')^{-1}([b,c]-[a,c]a^{-1}b)(d-ca^{-1}b)^{-1}$, obtaining at last
 	\begin{align*}
 		L_{21}&=(\Delta')^{-1}([a,c]a^{-1}(ac^{-1}d-b)-([b,c]-[a,c]a^{-1}b) )(ac^{-1}d-b)^{-1}\\
 		&=(\Delta')^{-1}([a,c]c^{-1}d-[b,c] )(ac^{-1}d-b)^{-1}\\
 		&=((ac^{-1}d-b)^{-1}-(\Delta')^{-1}c ).
 	\end{align*}
 	From the former equalities  one obtains the concrete form of the decomposition matrix $\mathbb{T}_L$ and therefore the left inverse matrix  explicitly, i.e.\ \begin{align*}
 	\mathbb{A'}^{-1}_L =\,_{c}\mathbb{A}^{-1}_L-\mathbb{T}_L.
 		\end{align*}
Next, we prove that the left inverses defined with the different orderings agree. Let us first write the inverse left matrix with the second ordering explicitly down, 
	\begin{align*}
\mathbb{A'}^{-1}_L&=(\Delta')^{-1}\begin{pmatrix} 
  \Delta' (a-bd^{-1}c)^{-1}    &-b-([d,b]-[d,a]a^{-1}b)(d-ca^{-1}b)^{-1} \\   -\Delta' (ac^{-1}d-b)^{-1}   &a-([b,c]-[a,c]a^{-1}b)(d-ca^{-1}b)^{-1}
\end{pmatrix}.
\end{align*}
 For the components $(\mathbb{A'}^{-1}_L)_{11}$ and $(\mathbb{A'}^{-1}_L)_{21}$ the equivalence to the Gel'fand inverse is straight forward while for the component $(\mathbb{A'}^{-1}_L)_{12}$
 we have 
 \begin{align*}
(\mathbb{A'}^{-1}_L)_{12}&=(\Delta')^{-1}\left( -b-([d,b]-[d,a]a^{-1}b)(d-ca^{-1}b)^{-1}\right)\\&=
-(\Delta')^{-1}\left( b(d-ca^{-1}b)+([d,b]-[d,a]a^{-1}b)\right)(d-ca^{-1}b)^{-1} \\&=
-(\Delta')^{-1}\left( -bc  +ad\right)a^{-1}b(d-ca^{-1}b)^{-1}\\&= 
- (db^{-1}a-c)^{-1}
 \end{align*}
 where in the last lines we multiplied by the inverse of the determinant and pulled the term $(a^{-1}b)$ in  to the bracket and for the component  $(\mathbb{A'}^{-1}_L)_{22}$ we analogously have
  \begin{align*}
 (\mathbb{A'}^{-1}_L)_{12}&=(\Delta')^{-1}\left(
 a-([b,c]-[a,c]a^{-1}b)(d-ca^{-1}b)^{-1}\right)\\&=
 (\Delta')^{-1}\left(
 a(d-ca^{-1}b)-([b,c]-[a,c]a^{-1}b)\right)(d-ca^{-1}b)^{-1}\\&= 
 (\Delta')^{-1}\left(
 ad - bc \right)(d-ca^{-1}b)^{-1}\\&=
 (d-ca^{-1}b)^{-1}
  \end{align*} 
  		 	\end{proof}
 Due to the form of the decomposition matrices we conclude that for the left inverse matrix the first choice of ordering w.r.t.\ the determinant (i.e.\ $\Delta$)  and in regards to the right inverse   the latter ordering of the determinant (i.e.\ $\Delta'$) has a more natural form. With more natural we mean that there is a natural ordering of the components of the  decomposition matrices such that they all contain commutators. Although   the inverse itself does not depend on the ordering, the particular  decomposition matrices  can be useful during effective calculations in noncommutative geometry in order  to produce higher orders in the deformation parameter. This in turn will prove to be technically valuable. 
 
\section*{Acknowledgements}
The authors would like to thank Marcos Rosenbaum and David Vergara for many fruitful discussions that lead to the idea developed here. One of the authors (AM) acknowledges partial support from the Conacyt project  $258259$.
\bibliographystyle{alpha}
\bibliography{CLAST1}
\end{document}